\documentclass[aps,prl,floatfix,superscriptaddress,showpacs,twocolumn,footinbib,longbibliography]{revtex4-1}

\usepackage{graphicx}
\usepackage{amsmath}
\usepackage{amssymb}
\usepackage{amsfonts}
\usepackage{amsthm}
\usepackage{mathtools}
\usepackage{siunitx}
\usepackage[unicode=true, breaklinks=true, pdfborder={0 0 1}, backref=false, colorlinks=true, linkcolor=blue, citecolor=blue, urlcolor=blue]{hyperref}
\usepackage{xcolor}
\usepackage{soul}
\usepackage{textcomp}
\usepackage{bm}
\usepackage{dsfont}
\usepackage{relsize}
\usepackage{braket}
\usepackage{enumitem}   
\usepackage{mathrsfs}
\usepackage{cleveref}
\usepackage{bigints}
\usepackage{framed}
\usepackage[makeroom]{cancel}
\usepackage{notes2bib}
\usepackage[draft]{changes} %draft,final

\theoremstyle{plain}

\numberwithin{obs}{section}

\definecolor{ms}{rgb}{0,.4,1}
\newcommand{\ms}[1]{{\color{black}#1}}
\definecolor{henrik}{rgb}{1,.4,0}

\definecolor{jacopino}{rgb}{0.65,0.42,1}

\newcommand{\Tr}[1]{\mathrm{Tr}\left[ #1\right]} %new

\newcommand{\bigo}[1]{\mathcal{O}\left (#1\right)}

\newcommand{\rom}[1]{\uppercase\expandafter{\romannumeral #1\relax}}

\newcommand{\sqrbra}[1]{\left[ #1\right]}

\newcommand{\J}{\mathbb{J}}

\usepackage{amssymb}             % AMS Math

\newcommand{\de}{{\rm d}}

\newcommand{\RR}{\mathbb{R}}
\newcommand{\LL}{\mathbb{L}}
\newcommand{\TrR}[2]{\mathrm{Tr}_{#1}\left[ #2\right]}
\newcommand{\idO}{\mathds{I}}
\newcommand{\norbraB}[1]{\boldsymbol{(}#1\boldsymbol{)}}

\newtheorem{theorem}{Theorem}
\newtheorem{corollary}{Corollary}[theorem]
\theoremstyle{definition}

\usepackage{algorithm}
\usepackage{algpseudocode}
\usepackage{comment}

\newtheorem*{theorem*}{Theorem}
\newtheorem*{corollary*}{Corollary}

\usepackage[normalem]{ulem}

\usepackage{tcolorbox}
\tcbuselibrary{breakable}
\tcbuselibrary{skins}
\definecolor{my-green}{RGB}{0,144,81}
\definecolor{my-red}{RGB}{255,113,91}
\tcbset{enhanced jigsaw, colback=my-green!10!white,colframe=my-green!85!black,fonttitle=\bfseries}

\begin{document}
	
	\title{ 
		Physicality of evolution and statistical contractivity are equivalent notions of maps
	}
	%\author{Team T.I.M.}
    %\affiliation{Sempre nel cuore}
    \author{Matteo Scandi} \email{matteo@ifisc.uib-csic.es}
    \affiliation{Institute for Cross-Disciplinary Physics and Complex Systems (IFISC) UIB-CSIC, Campus Universitat Illes Balears,E-07122 Palma de Mallorca, Spain}
	\author{Paolo Abiuso} \email{paolo.abiuso@oeaw.ac.at}
    \affiliation{Institute for Quantum Optics and Quantum Information - IQOQI Vienna, Austrian Academy of Sciences, Boltzmanngasse 3, A-1090 Vienna, Austria}
	\author{Dario De Santis}\email{dario.desantis@sns.it}
    \affiliation{Scuola Normale Superiore, Piazza dei Cavalieri 7, I-56126 Pisa, Italy}
	\author{Jacopo Surace} \email{jsurace@perimeterinstitute.ca}%\Instagram{@jacopo\_sorace}
    \affiliation{Perimeter Institute for Theoretical Physics, 31 Caroline Street North, Waterloo, ON N2L 2Y5, Canada}

	\date{\today}
	
	\begin{abstract}
		Statistical quantifiers are generically required to contract under physical maps, following the intuition that information should be lost under noisy transformations. This principle is so important in statistics that it even allows to derive uniqueness results based on it:
		the Chentsov-Petz theorem identifies the Fisher information metrics as the only family, on the space of probability distributions (or density matrices), that contracts under physical maps. This construction could suggest that statistical quantifiers are a derived concept, while the only fundamental object are the physical evolutions. The aim of this work is to disproof this belief. To this end, we prove a statement dual to the Chentsov-Petz theorem, showing that among all possible linear maps, the only ones that contract the Fisher information are exactly the physical ones. This result proves that, contrary to the common opinion, there is no fundamental hierarchy between physical maps and canonical statistical quantifiers, as either of them can be defined in terms of the other.
	\end{abstract}

	\maketitle
	
	\section{Introduction}
	
        The result that gave birth to information theory, Shannon's noiseless coding theorem, answers the question of what is the maximum rate at which a message can be transmitted over an ideal transmission line~\cite{Shannon1948,Holevo2013}. This can be phrased in terms of the relative entropy, defined for two discrete probability vectors $\boldsymbol{p}$ and $\boldsymbol{q}$ as:
	\begin{align}
		D(\boldsymbol{p}||\boldsymbol{q}) := \sum_{i}\; p_i\,\log \frac{p_i}{q_i}\,,\label{eq:relEnt}
	\end{align}
        where $p_i/q_i$ are the components of the two vectors. Then, the minimum compression is connected to the statistical difference, as measured by $D(\boldsymbol{p}||\boldsymbol{q})$, between the code probability distribution and the uniform one. Thanks to the operational interpretation of the relative entropy as the asymptotic %behaviour of the 
        difference between the two log-likelihood~\cite{kullbackInformationTheoryStatistics1997}, this result simply means that in order for a probabilistic code to work, it needs to generate correlations significantly different from random bits.

    Over time, information theory evolved from its role in telecommunication to become one of the cornerstones of the modern understanding of physics. Indeed, apart from its decisive role in the treatment of statistical mechanics~\cite{jaynes1957information,leff2014maxwell,parrondo_thermodynamics_2015}, its extension to the quantum regime gave rise to an entire new field of physics~\cite{nielsen2002quantum}. Today, the concept of information plays a foundational role very similar to the one of energy  for classical mechanics. For example, time-reversal symmetry follows from the requirement that the global evolution should conserve information, constraining the dynamics, thanks to Wigner's theorem~\cite{Holevo2013}, to unitaries or antiunitaries transformations.
	Moreover, much in the same way in which one reconciles the existence of dissipative dynamics with the principle of conservation of energy, we explain the existence of evolutions for which information appears to be lost by postulating that the missing information is simply transferred to some environmental degrees of freedom which we do not have access to.  Following this argument, the most generic quantum mechanical evolution $\Phi$ takes the form:
	\begin{align}
		\Phi(\rho_S) = \TrR{E}{U_{SE}\,(\rho_S\otimes\omega_E) \,U_{SE}^\dagger}\,,\label{eq:CPTPdef}
	\end{align}
	where $\rho_S$ is the state of the system, $\omega_E$ is a generic environmental state and $U_{SE}$ is the global unitary/antiunitary evolution. 
  %\sout{The class of maps defined by Eq.~\eqref{eq:CPTPdef} is characterised by Stinespring theorem~\cite{stinespring1955positive} as the set of Positive Trace-Preserving maps (PTP) \paolo{(è vero questo?)} \js{(non ho trovato nessuna equivalenza CP$\leftrightarrow$forma (2).)} \js{(Volendo ci si può pensare e mettere una prova in appendice)} \js{(Pero sicuro non è scritta nello Stinespring)}. Moreover, requiring the compatibility with composition of different systems, further restricts the set of global transformations $U_{SE}$ to unitaries only~\REF, meaning that the ones in Eq.~\eqref{eq:CPTPdef} are Completely Positive maps (CP). }
Moreover, requiring the compatibility with composition of different systems, further restricts the set of global transformations $U_{SE}$ to unitaries only~\cite{zyczkowski2004duality,chiribella2021symmetries}. The class of maps defined by Eq.~\eqref{eq:CPTPdef}, with $U_{SE}$ restricted to be unitary,  is characterised by Stinespring theorem~\cite{stinespring1955positive} as the set of  Completely Positive Trace-Preserving maps (CPTP). 
When the same reasoning is applied to classical probabilities, one obtains the set of stochastic maps, since these can always be decomposed as:
		\begin{align}
			[T(\boldsymbol{p})]_i=\sum_{j,k,l} U_{ij}^{kl}p_k\omega_l\;,
		\end{align}
		where $\boldsymbol{p}$  represents the classical probability distribution of the system, $\boldsymbol{\omega}$ the one of the environment, while $U$ is a global isometry. %Interestingly, in this case the concept of positive and completely positive maps coincide.
	
		Despite their different origins, statistical quantifiers and dynamical maps are deeply connected. For one thing, the relative entropy in Eq.~\eqref{eq:relEnt} (as well as its quantum counterpart $H_L(\rho||\sigma):=\Tr{\rho(\log\rho-\log\sigma)}$)  
        contracts under physical evolutions:% i.e., in formulae:
		\begin{align}
			D(\boldsymbol{p}||\boldsymbol{q}) \geq D(T(\boldsymbol{p})||T(\boldsymbol{q})) \,,\label{contrD}
		\end{align}
		for any stochastic transformation $T$. Similarly, any sensible generalisation of the relative entropy needs to satisfy a condition akin to Eq.~\eqref{contrD}, 
		corresponding to the intuition that information cannot increase during noisy transformations. This connection becomes even more compelling when one looks at the local behaviour of statistical quantifiers:  a renown theorem by Chentsov establishes that there exists a unique metric on the space of classical probability distributions  that is contractive under any stochastic (i.e., physical) map~\cite{cencovStatisticalDecisionRules2000,campbell1986extended}. This takes the name of Fisher information metric, or Fisher-Rao metric, and it is \ms{the main object of interest of this work}. In this way, dynamical considerations constrain the local behaviour of any possible statistical quantifier.
		The deep connection between dynamics and Fisher information~\cite{abiuso2023characterizing,scandi2023quantum} has been mostly overlooked in the literature, as more operative uses were explored~\cite{Cramer46, Rao92,ly_tutorial_2017}. 
  
        Still, there is another path to the Fisher metric that we want to highlight here. In an effort to gather the many possible statistical quantifiers present in the literature, Csiszár introduced the concept of \ms{generalised information divergences~\footnote{Also known as $f$-divergences.}}, %{$f$-divergences}, 
        a continuous family generalising many of the main properties of the relative entropy~\cite{csiszarInformationTheoryStatistics2004}. Interestingly though, despite the wide range of possible definitions, it was noticed that for close-by probabilities (in the sense that $\boldsymbol{q} = \boldsymbol{p}+\de \boldsymbol{q}$ for some $||\de \boldsymbol{q}||\ll 1$) all the %$f$
        \ms{Csiszár}-divergences reduce to the Fisher information. Indeed, if we take for example the relative entropy in Eq.~\eqref{eq:relEnt} it is a matter of Taylor-expanding the logarithm to see that:
	\begin{align}
		D(\boldsymbol{p}||\boldsymbol{p}+\de \boldsymbol{q}) \simeq \sum_i\;\frac{\de q_i^2}{2p_i} = \frac{1}{2}\sum_{i,j}\;\de q_i\,\eta_{i,j}\big|_{\boldsymbol{p}}\,\de q_j\,,
		\label{eq:cl_fi_me}
	\end{align}
	\ms{where  $\eta_{i,j}\big|_{\boldsymbol{p}} := \frac{\delta_{i,j}}{p_i}$ is exactly the Fisher metric, and we do not consider} higher order corrections in $\de \boldsymbol{q}$. This shows that the Fisher information metric could also be defined in purely statistical terms, that is, independently of any dynamical considerations.
	
	Indeed, even if one might think that there is a hierarchy between physical maps and Fisher information (and, therefore, statistical quantifiers), in which the latter derives from the first, we show in this work that the situation is actually more intertwined: 
    we prove that physical evolutions can be defined solely in terms of the Fisher information metric. Before doing this, though, we need to introduce the concept of contrast functions and Fisher information metrics for quantum states.
    
	\section{Quantum Fisher Information metrics}
	In this section we give a brief overview of the quantum Fisher information metrics, presenting in particular their statistical derivation (Thm.~\ref{cf:thm:Ruskai}) and the dynamical one (Thm.~\ref{cf:thm:Petz}). \ms{For the sake of brevity,} we will only outline the main definitions, and we refer the reader to the review in~\cite{scandi2023quantum} for further details.
	
	Let us then introduce the concept of contrast functions $H(\rho||\sigma)$ between two quantum states $\rho$ and $\sigma$. This notion was first introduced for classical statistics by Csiszár~\cite{csiszarInformationTheoryStatistics2004}, and later extended by Petz to the quantum regime~\cite{petz1986quasi}, following an effort to give an axiomatic foundation to the many different divergences present in the literature. In particular, the principles proposed are the following:
    \ms{
    \emph{1) positivity and faithfulness}:  
    $H(\rho||\sigma)\geq0$, with equality \emph{iff} $\rho\equiv\sigma$;
    \emph{2) joint convexity}: mixing states should not increase their information content, therefore ${H(\lambda \rho_1 +(1-\lambda )\rho_2||\lambda \sigma_1 +(1-\lambda )\sigma_2)\;
        \leq}\allowbreak
        {\lambda H(\rho_1 ||\sigma_1 ) + (1-\lambda) H(\rho_2 ||\sigma_2 )}$,
		for $0\leq\lambda\leq1$;
   \emph{3) monotonicity}: an analogous condition for compatibility with physical evolutions $H(\Phi(\rho)||\Phi(\sigma))\leq H(\rho||\sigma)$;
   \emph{4) differentiability}: $C^\infty$-regularity in the arguments of $H$.
    }
	These conditions are still not enough to single out a unique family of statistical quantifiers, and for this reason it is customary to also propose an ansatz \ms{for $H(\rho||\sigma)$}, which encompasses most of the examples defined in the literature. In particular, the one proposed in~\cite{petz1986quasi} reads as follows:
	\begin{align}\label{cf:eq:HGexpressionApp}
		H_g (\rho||\sigma) := \Tr{ g\boldsymbol{(}\LL_\sigma\RR_\rho^{-1}\boldsymbol{)}\sqrbra{\rho}}\,,
	\end{align}
	where $g:\RR^+\rightarrow\RR$ is a matrix convex function~\cite{hiaiIntroductionMatrixAnalysis2014}, and $\LL_\rho$ and $\RR_\rho$ are the left and right multiplication operators, acting as:
	\begin{align}
		\LL_\rho[A] = \rho\,A\,,\qquad\qquad
		\RR_\rho[A] = A\,\rho\,.
	\end{align}
	It should be noticed that despite being a very large family, containing the relative entropy (when $g(x)=-\log x$), and functionals related to the Rényi divergences, the definition in Eq.~\eqref{cf:eq:HGexpressionApp} is still not enough to completely exhaust all the possibilities present in the literature~\cite{scandi2023quantum}. Still, it is enough to define the entire family of Fisher information metrics. Indeed, we have the following:
	\begin{theorem}[Lesniewski, Ruskai~\cite{lesniewskiMonotoneRiemannianMetrics1999}]\label{cf:thm:Ruskai} The contrast functions $H_g(\pi + \varepsilon A || \pi+ \varepsilon B) $ can be locally approximated up to corrections of order $\bigo{\varepsilon^3}$ as:
		\begin{align}
			H_g(\pi + \varepsilon A || \pi+ \varepsilon B)  \simeq\frac{\varepsilon^2}{2}\, \Tr{(A-B)\,\J_f^{-1}\big|_{\pi} [(A-B)]}\,,\label{cf:eq:thm2HGexp}
		\end{align}
		where $A$ and $B$ are traceless, Hermitian perturbations, and the superoperator $\J_f\big|_\pi $, called the quantum Fisher operator, is defined as:
		\begin{align}
			\J_f\big|_\pi := \RR_\pi \,f\norbraB{\LL_\pi\RR_\pi^{-1}}\,.\label{cf:eq:65}
		\end{align}
		Moreover, $f$ is an operator monotone function connected to  $g$ by the equation:
		\begin{align}\label{cf:eq:correspondenceFG}
			f(x) = \frac{(x-1)^2}{g(x) + x \, g(x^{-1})}\,.
		\end{align}
	\end{theorem}
	This represents the first possible derivation of the Fisher information metrics for quantum systems. Still,	in order to claim that the operators in Eq.~\eqref{cf:eq:65} define all and only the Fisher information metrics we need to resort to their dynamical characterisation, i.e., their contractivity under arbitrary CPTP maps. This %is extra bit 
	is provided by the independent result of Petz:% following:
	\begin{theorem}[Petz~\cite{petzMonotoneMetricsMatrix1996}]\label{cf:thm:Petz}
		All the metrics on quantum states that contract under arbitrary CPTP maps are of the form:
		\begin{align}
			K_{f,\pi}(A,B)	 := \Tr{A\, \,\J_f^{-1}\big|_\pi [B]},\label{cf:eq:monotoneMetrics}
		\end{align}
		where $f:\RR^+\rightarrow\RR^+$ is an operator monotone function from the family in Eq.~\eqref{cf:eq:correspondenceFG}. 
	\end{theorem}
	This result makes the definition of Fisher information metrics in the quantum regime unambiguous. Indeed, even if the single metric $\eta_{i,j}\big|_{\boldsymbol{p}}$ is now replaced by a continuous family of possible operators $\J^{-1}_f\big|_\pi $, each one describing the local behaviour of a different contrast function (Thm.~\ref{cf:thm:Ruskai}), it would not be correct to postulate some extra axioms to isolate a unique quantity. Indeed, Thm.~\ref{cf:thm:Petz} tells us that narrowing down the family in Eq.~\eqref{cf:eq:65} would also destroy the connection with the dynamical characterisation of the Fisher information metrics. \ms{Several metrics in this family have important applications in information theory, such as 
    the Bures metric (corresponding to $\J_{\frac{1+x}{2}}\big|_\pi[A]=\frac{1}{2}\{\pi,A\}$) in quantum metrology~\cite{paris2009quantum}, %via the Quantum Cramér-Rao bound\REF, 
    the Kubo-Mori-Bogoliubov metric ($\J_{\frac{x-1}{\ln x}}\big|_\pi[A] = \int_0^1{\rm d}s\,  \pi^s A\pi^{(1-s)}$) as the local expansion of the quantum relative entropy, or
    the Wigner-Yanase metric ($\J_{\frac{(1+\sqrt{x})^2}{4}}\big|_\pi[A]=\frac{1}{4}\{\sqrt{\pi},\{\sqrt{\pi},A\}\}$) in hypothesis testing~\cite{audenaert2007discriminating}. %via the Quantum Chernoff bound
    For a more exhaustive classification, see~\cite{scandi2023quantum}}.
    \ms{Despite the exotic form of these superoperators, % in Eq.~\eqref{cf:eq:65},
    it should be noticed that for commuting observables the uniqueness of the classical Fisher information is reattained. That is, if $[A, \pi] = 0$, then:
	$
		\J_f^{-1}\big|_\pi [A] = \pi^{-1} \,A\;,%\label{eq:classicalFisher}
	$
    irrespectively of the defining function $f$.}
    
    %\ms{These metrics found applications in a wide range of problems, from metrology, to hypothesis testing, to thermodynamics~\REF. Still, their dynamical nature was for partially overlooked in the literature. The next section aims at filling this gap.}

	\section{Contractivity implies physicality}
    In this section we prove our main result, that is how the property of statistical contractivity isolates physical evolutions among all possible linear maps. Our main mathematical theorem (Thm.~\ref{cf:theo:PiffContracts}) characterizes positive maps via their contraction of the Fisher information metric. This can be then extended to complete positivity (Cor.~\ref{cor:1}) and contrast functions contractivity (Cor.~\ref{cor:2}).

    \begin{figure}
		\centering
		\includegraphics[width=1.0\linewidth]{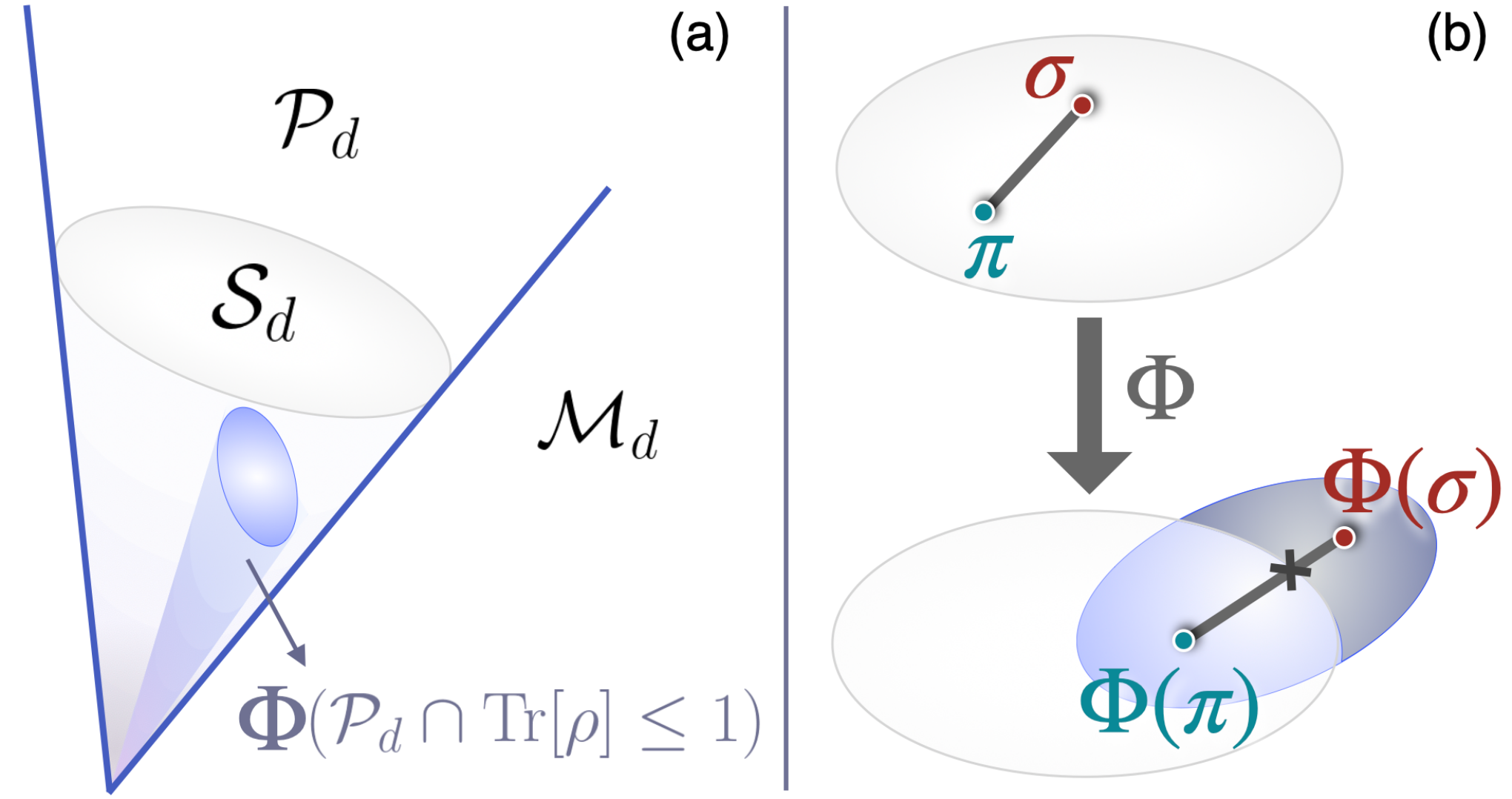}
			\caption{\ms{We consider a generic linear map $\Phi$, and its action on the cone of positive operators $\mathcal{P}_d$ (including normalized states $\mathcal{S}_d$ with trace 1).
            (a) Depiction of our main result (Thm.~\ref{cf:theo:PiffContracts}): the contractivity of the Fisher information metric~\eqref{eq:contractionContrast} implies that $\Phi$ is positive ($\Phi(\mathcal{P}_d) \subseteq \mathcal{P}_d$) and does not increase the trace.
            (b) Sketch of the positivity proof, valid also when restricting to $\mathcal{S}_d$ only (Cor.~\ref{cor:1}):} thanks to condition~\ref{cf:it:2} one can always find a point $\pi$ in $\mathcal{P}^\circ_d$ that gets mapped into $\mathcal{P}^\circ_d$. Now assume there is a matrix $\sigma$ in the $\mathcal{P}^\circ_d$ that gets mapped outside of $\mathcal{P}_d$. Then, the segment connecting $\pi$ to $\sigma$ will be mapped into a curve intersecting the boundary of $\mathcal{P}_d$. Since the Fisher information is bounded on the interior of the space of states, and diverges on its boundary, there must be a point for which the Fisher metric increases.}
		\label{fig:proof}
	\end{figure}
    
	%We are now ready to prove our main result, a characterisation of physical evolutions among the family of all linear maps. 
    Define $\mathcal{M}_d(\mathbb{C})$ to be the space of $d \times d$ complex matrices,  $\mathcal{P}_d\subset\mathcal{M}_d(\mathbb{C})$ to be the subset of positive semidefinite matrices, and $\mathcal{S}_d\subset\mathcal{P}_d$ the space of states, i.e., the matrices in $\mathcal{P}_d$ with trace one. Moreover, denote the interior of a set $X$ as $X^\circ$. Then, it holds that:
	\begin{theorem}\label{cf:theo:PiffContracts}
		Consider a Hermitian preserving, linear map $\Phi: \mathcal{M}_d(\mathbb{C})\rightarrow\mathcal{M}_d(\mathbb{C})$\ms{, and any choice $f$ of Fisher information metrics~\eqref{cf:eq:monotoneMetrics}}.  If $\Phi$ satisfies the following two properties: 
		\begin{enumerate}
			\item $\Phi$ maps at least one point from $ \mathcal{P}^\circ_d$ into $\mathcal{P}^\circ_d$;\label{cf:it:2}
			\item %for any matrix $\rho$ in $\mathcal{P}^\circ_d\cap\norbraB{\Phi^{-1}\norbraB{ \mathcal{P}^\circ_d\cap\Phi( \mathcal{P}_d)}}$ 
            for any matrix $\rho \in \mathcal{P}^\circ_d$ such that $\Phi(\rho)\in \mathcal{P}^\circ_d$, and any vector $\delta \rho$ in the $\mathcal{P}_d$-tangent space, \ms{the Fisher information decreases:}%one has:
            \label{cf:it:3}
			\begin{align}
				K_{f,\rho}(\delta\rho,\delta\rho) \geq
				K_{f,\Phi(\rho)}(\Phi(\delta\rho),\Phi(\delta\rho)) \,;\label{eq:contractionContrast}
			\end{align}
		\end{enumerate}
		then the image of ${\Phi}$ is completely contained in $\mathcal{P}_d$, meaning that $\Phi$ is a positive map (P). Moreover, $\Phi$ is trace non-increasing.
	\end{theorem}
	We refer the reader to  Appendix~1 for the precise proof of the theorem. The main idea is explained in Fig.~\ref{fig:proof}, and it leverages the fact that the Fisher information metric is bounded in the interior of $\mathcal{P}_d$, while it diverges on its boundary. Then, one can prove the claim by contradiction: if there is even a single point in $\mathcal{P}_d^\circ$ that gets mapped outside of $\mathcal{P}_d$, thanks to the linearity of $\Phi$, one can define a segment in $\mathcal{P}_d^\circ$ whose image will intersect the boundary of $\mathcal{P}_d$ (see Fig.~\ref{fig:proof}). This allows to find a point for which the Fisher metric expands, contradicting condition~\ref{cf:it:3} and proving the claim.
	
	There are a couple of remarks to be made: first, it should be noticed that since the Fisher information metrics are only defined on $\mathcal{P}_d^\circ$, condition~\ref{cf:it:2} is minimal to introduce the concept of contractivity (encoded by Eq.~\eqref{eq:contractionContrast}). Moreover, the restriction on the set of $\rho$ considered at the beginning of condition~\ref{cf:it:3} is also minimal in order for the two sides of Eq.~\eqref{eq:contractionContrast} to be regular, i.e., we only focus on  matrices $\rho$ for which both the Fisher metrics in $\rho$ and in $\Phi(\rho)$ are well-defined. %\deleted{For example, in Fig.~\ref{fig:proof} $\mathcal{P}^\circ_d\cap\norbraB{\Phi^{-1}\norbraB{ \mathcal{P}^\circ_d\cap\Phi( \mathcal{P}_d)}}$ would correspond to the preimage of the lighter set on the right}.
    Finally, let us notice that the use of $\mathcal{P}_d$ as a domain was chosen to show the mathematical result in full generality, as well as to accommodate potentially trace non-preserving maps. In this sense, the theorem again isolates physical maps, which can only \emph{decrease} the probability-normalization of states~\footnote{Maps that decrease the trace correspond to evolutions with post-selection, i.e., with a non-null probability of discarding the system.}.  
    In the corollaries that follows we restrict, for clarity, to standard trace-preserving (TP) maps on the physical set of states $\mathcal{S}_d$.

    \begin{comment}
	\begin{corollary}\label{cor:1}
		Consider a Hermitian preserving, trace-preserving linear map $\Phi: \mathcal{M}_d(\mathbb{C})\rightarrow\mathcal{M}_d(\mathbb{C})$. Assume that $\Phi$: maps at least one point from $ \mathcal{S}^\circ_d$ into $\mathcal{S}^\circ_d$ and 2) for all $\rho\in\mathcal{S}^\circ_d$ such that $\Phi(\rho)\in \mathcal{S}^\circ_d$, it contracts the Fisher metric as in Eq.~\eqref{eq:contractionContrast} for all $ \delta \rho$ in the $\mathcal{S}_d$-tangent space. Then $\Phi$ is a positive map (PTP).
  
        If $\Phi\otimes\idO_d$, acting on  $(\mathcal{S}_d\otimes\mathcal{S}_d)^\circ$, satisfies the same conditions, % of Thm.~\ref{cf:theo:PiffContracts},
        then $\Phi$ is a completely positive map (CPTP).
	\end{corollary}
    \end{comment}
    \begin{corollary}\label{cor:1}
		Consider a Hermitian preserving, trace-preserving linear map $\Phi: \mathcal{M}_d(\mathbb{C})\rightarrow\mathcal{M}_d(\mathbb{C})$. If $\Phi$ satisfies the following two properties:
        \begin{enumerate}
            \item it maps at least one point from $ \mathcal{S}^\circ_d$ into $\mathcal{S}^\circ_d$;
            \item for all $\rho\in\mathcal{S}^\circ_d$ such that $\Phi(\rho)\in \mathcal{S}^\circ_d$, and for all $ \delta \rho$ in the $\mathcal{S}_d$-tangent space, it contracts the Fisher metric as in Eq.~\eqref{eq:contractionContrast};
        \end{enumerate}
        Then $\Phi$ is a positive map (PTP). Moreover, if $\Phi\otimes\idO_d$, acting on  $(\mathcal{S}_d\otimes\mathcal{S}_d)^\circ$, satisfies the same conditions, % of Thm.~\ref{cf:theo:PiffContracts},
        then $\Phi$ is a completely positive map (CPTP).
	\end{corollary} 
    The proof of the corollary is contained at the end of Appendix 1, and it is equivalent to the one for Thm.~\ref{cf:theo:PiffContracts}. It should be noticed that Thm.~\ref{cf:theo:PiffContracts} and Cor.~\ref{cor:1} completely characterize physical evolutions: the first gives a condition for positivity, which is sufficient for classical systems, while the latter characterise CPTP maps, the relevant notion of maps for quantum systems.
    %for quantum states in the standard bipartite CPTP framework.  For classical systems, where the two concepts coincide, positivity  (Thm.~\ref{cf:theo:PiffContracts}) is instead sufficient to physicality.
	%\deleted{Thanks to the identification between complete positivity and physicality of evolutions, this results shows that the contractivity of the Fisher information can be used as the definition of physical dynamics. Moreover, it should also be noticed that once again this result is minimal, in the sense that one needs to consider the ancillary dynamics $\Phi\otimes\idO_d $ in order to distinguish between positivity of a map and complete positivity \paolo{(falso, nel caso in cui puoi fare tomografia)}. Indeed, for classical systems, where these two concepts coincide, Thm.~\ref{cf:theo:PiffContracts} is sufficient to prove physicality.}
 
	We conclude by observing that the same results can be phrased in terms of \emph{any of the contrast functions} introduced above.
    %with one last result about contrast functions. 
    In fact, putting together Thm.~\ref{cf:thm:Ruskai} with Cor.~\ref{cor:1}, we have:
	\begin{corollary}\label{cor:2}
		Consider a  Hermitian preserving, trace preserving, linear map $\Phi: \mathcal{M}_d(\mathbb{C})\rightarrow\mathcal{M}_d(\mathbb{C})$ mapping at least one point from $ \mathcal{S}^\circ_d$ into $\mathcal{S}^\circ_d$. Consider any contrast function $H_g (\rho||\sigma) $. If for every two states $\rho$ and $\sigma$ in $\mathcal{S}_d^o$ it holds that:
		\begin{align}
			H_g (\rho||\sigma) \geq H_g (\Phi(\rho)||\Phi(\sigma))\,,
		\end{align}
		then $\Phi$ is a positive map (PTP). Moreover, if the same holds for $\Phi\otimes\idO_d$ then the map is completely positive (CPTP).
	\end{corollary}
 
	We refer to Appendix~2 for the precise proof, but this result simply follows from the fact that all contrast functions locally expand to Fisher metrics, hence it is sufficient to choose $\rho$ and $\sigma$ to be close-by matrices and then apply Thm.~\ref{cf:theo:PiffContracts} to them.
	Cor.~\ref{cor:2} shows that it is enough to consider any divergence to actually characterise the class of positive (or completely positive) maps. %This raises the question of why one would prefer to use the Fisher information metrics, a rather uncommon family of quantities, compared to the most standard among the divergences, namely the relative entropy. 
    In fact, while our results are more general, a specific instance of Cor.~\ref{cor:2} can be expressed as follows:
    \begin{center}
    \emph{If a linear map acting on the set of states (including, in the quantum case, ancillas with trivial dynamics) cannot increase their relative entropy, then such map represents a physical evolution.}
    \end{center}
 
    %\deleted{In this context, it should be noticed that while Thm.~\ref{cf:theo:PiffContracts} can still be formulated using any statistical quantifier other than the Fisher information, only the latter can be characterised in the form of Thm.~\ref{cf:thm:Petz}. Thus, whereas one would retain the statistical definition of completely positive maps, it would then lose the equivalence between dynamics and statistics. This once again demonstrates the foundational importance of the Fisher information.}
	
	\section{Conclusion}
        In the axiomatic definition of quantum statistical divergences, the monotonicity condition stands out as the only request which is not purely statistical in nature: indeed, whereas the other properties can be expressed solely in terms of states, in order to introduce monotonicity one needs to explicitly introduce the concept of noisy transformations. This fact might suggest that there is some hierarchy between physical evolutions and statistical quantifiers, since one needs the first to define the latter. Corroborating this belief, the Chentsov-Petz theorem (Thm.~\ref{cf:thm:Petz}) constrains the local behaviour of statistical distances, singling out the Fisher information through a purely dynamical request.

        Still, historically most of the contrast functions were introduced in a purely statistical setting, therefore their strong connection with dynamics is \emph{a priori} not obvious and, indeed, quite remarkable. This manuscript, then, aims at clarifying this fortunate conjunction: Thm.~\ref{cf:theo:PiffContracts} and its Cor.~\ref{cor:1} show that physical evolutions can be singled out as the unique family of linear maps that contract the Fisher information. This result not only justifies the claim that physicality of evolution and statistical contractivity are equivalent notions of maps, but also hints at the source of this connection, as both dynamical maps and statistical distances arise from the underlying notion of conservation of information.
        
        Moreover, Cor.~\ref{cor:2} also shows how one could in principle define physical evolutions merely in terms of any contrast function, strengthening the message that there is no intrinsic hierarchy of fundamentality between statistical notions and dynamical ones. Still, it should be kept in mind that Thm.~\ref{cf:thm:Petz} is peculiar to the Fisher information, further corroborating the natural role this quantity plays both in quantum and classical physics.
        
        In fact, it was noticed that many dynamical aspects of evolutions, as Markovianity or detailed balance, can be naturally phrased in terms of the behaviour of maps with respect to the Fisher metrics~\cite{abiuso2023characterizing,scandi2023quantum}. The results presented here suggest that this is far from being a coincidence, showing that there is a fundamental connection which is worth exploring further in future works. Moreover, in the context of the reconstructions of quantum mechanics~\cite{hardy2015reconstructing,mueller2016information}, the fact that physical evolutions can be singled out in terms of a statistical quantity as the Fisher information might be of particular relevance for the derivation of a sensible notion of dynamics.

    \section{Acknowledgements}
    M.S. acknowledges the support of project Intramural 20235AT009 of the Spanish National Research Council. P.A. is supported by the QuantERA II programme, that has received funding from the European Union’s Horizon 2020 research and innovation programme under Grant Agreement No 101017733, and from the Austrian Science Fund (FWF), project I-6004.
    Research at the Perimeter Institute for Theoretical Physics is supported by the Government of Canada through the Department of Innovation, Science and Economic Development Canada and by the Province of Ontario through the Ministry of Research, Innovation and Science. D.D.S. is supported by the research project “Dynamics and Information Research Institute - Quantum Information, Quantum Technologies” within the agreement between UniCredit Bank and Scuola Normale Superiore di Pisa
(CI14\_UNICREDIT\_MARMI).
 
	\bibliography{bib}
	
	\newpage \ \newpage
	\onecolumngrid
	\appendix

    \section{Methods}
    
	\section{1. Proof of Thm.~\ref{cf:theo:PiffContracts} and Cor.~\ref{cor:1}}\label{app:1}
	
	In this appendix we present the proof of Thm.~\ref{cf:theo:PiffContracts}, which we repeat here for convenience:
	\begin{theorem*}
		Consider a Hermitian preserving, linear map $\Phi: \mathcal{M}_d(\mathbb{C})\rightarrow\mathcal{M}_d(\mathbb{C})$.  If $\Phi$ satisfies the following two properties: 
		\begin{enumerate}
			\item $\Phi$ maps at least one point from $ \mathcal{P}^\circ_d$ into $\mathcal{P}^\circ_d$;\label{cf:it:22}
			\item for any matrix $\rho$ in $\mathcal{P}^\circ_d\cap\norbraB{\Phi^{-1}\norbraB{ \mathcal{P}^\circ_d\cap\Phi( \mathcal{P}_d)}}$, and any tangent vector $\delta \rho$, one has:\label{cf:it:32}
			\begin{align}
				K_{f,\rho}(\delta\rho,\delta\rho) \geq
				K_{f,\Phi(\rho)}(\Phi(\delta\rho),\Phi(\delta\rho)) \,;\label{eq:contractionContrast2}
			\end{align}
		\end{enumerate}
		then the image of ${\Phi}$ is completely contained in $\mathcal{P}_d$, meaning that $\Phi$ is a positive map (P). Moreover, $\Phi$ is trace non-increasing.
	\end{theorem*}
	
	\begin{proof}
		
		We prove Theorem~\ref{cf:theo:PiffContracts}
		by contradiction: suppose there exists a map $\Phi$ that is not positive, but that satisfies Eq.~\eqref{eq:contractionContrast2}. Thanks to condition~\ref{cf:it:22} there exists at least one point $\pi$ in $\mathcal{P}^\circ_d$ such that its evolution is also in the interior of $\mathcal{P}_d$. Moreover, from the assumption that $\Phi$ is not P there is also at least one matrix $\sigma$ such that ${\Phi}(\sigma)\notin \mathcal{P}_d$. Without loss of generality one can choose $\sigma$ to be in the interior of $ \mathcal{P}_d$: if this is not the case we take a ball around $\sigma$ small enough so that its image still lays outside of the state space. Then, by inspecting the intersection between the preimage of this ball and $ \mathcal{P}^\circ_d$ one can find a point satisfying the assumption. Since both $\rho$ and $\sigma$ are in $\mathcal{P}^\circ_d$, the line $\rho_\lambda := (1-\lambda)\,\pi + \lambda \,\sigma$ is completely contained in $\mathcal{P}^\circ_d$ for $\lambda\in[0,1]$. This implies that the following superior is finite:
		\begin{align}
			\sup_{\lambda, \Tr{\delta\rho^2} = 1}&\,K_{f,\rho_\lambda}(\delta\rho,\delta\rho)  = \sup_{\lambda, \Tr{\delta\rho^2} = 1} \, \Tr{\delta\rho\,\J_f^{-1}\big|_{\rho_\lambda} [\delta\rho]} < \infty\,,
		\end{align}
		since the Fisher information is a bounded operator when restricted to a closed set completely inside of $\mathcal{P}_d^\circ$ (see~\cite{scandi2023quantum}). 
		By varying $\lambda$, ${\Phi}(\rho_{\lambda})$ interpolates linearly between the positive definite matrix ${\Phi}(\pi)$ and one with at least one negative eigenvalue, namely ${\Phi}(\sigma)$. Then, there exists a $\lambda^*$ such that ${\Phi}(\rho_{\lambda^*})$ is a matrix with at least one zero-eigenvalue. We are now ready to prove the claim. Set the state $\rho_\eta$ so that the smallest eigenvalue $\epsilon_\eta$ of ${\Phi}(\rho_\eta)$ is of order $\eta$, where $\eta\ll 1$. Call the associated eigenvector $\ket{\psi_\eta}$. Then, choose a perturbation $\delta \rho_\eta$ such that, for any $\eta$, $\Phi(\delta \rho_\eta)$ has a positive and finite contribution along the eigenvectors corresponding to $\eta$-eigenvalues (the positivity condition ensures that in the limit $\eta\rightarrow0$ the perturbed state is still in the interior of $ \mathcal{P}^\circ_d$ for any finite $\eta$, allowing for the definition of the Fisher information). One can always find such a perturbation. Suppose, in fact, the contrary. Then, this would mean that for any $\delta\rho$, it holds that:
		\begin{align}
			\lim_{\eta\rightarrow 0} \bra{\psi_\eta}\Phi(\delta\rho)\ket{\psi_\eta}=0 \quad \forall\,\delta\rho.\label{eq:20}
		\end{align}
		But any $\delta\rho$ can always be rewritten as $\delta\rho=\sigma_1-\sigma_2$, where $\sigma_{1/2}\in\mathcal{P}_d$. Then, Eq.~\eqref{eq:20} would imply that for any two matrices in $\mathcal{P}_d$ it holds that:
		\begin{align}
			\bra{\psi_0}\Phi(\sigma_1)\ket{\psi_0}=\bra{\psi_0}\Phi(\sigma_2)\ket{\psi_0} \quad \forall\sigma_1,\sigma_2\in \mathcal{P}_d\;.\label{eq:17A}
		\end{align}
		Set $\sigma_2:=\rho_{\lambda^*}$,  which, by definition, satisfies $\bra{\psi_0}\Phi(\rho_{\lambda^*})\ket{\psi_0}=0$. This implies that for any matrix in $\sigma\in\mathcal{P}_d$, the image of $\sigma$ is not full-rank, since $\bra{\psi_0}\Phi(\sigma)\ket{\psi_0}=0 $, which in turns contradicts the hypothesis that at least one point in $\mathcal{P}^\circ_d$ is mapped inside $\mathcal{P}^\circ_d$. Thus, there exists at least one $\delta\rho_\eta$ for which Eq.~\eqref{eq:20} does not hold. Then, choosing such a vector, one can see that
		the evolved Fisher information is unbounded as $\eta\rightarrow0$ by explicit inspection of its coordinates expression:
		\begin{align}
			K_{f,{\Phi}(\rho_\eta)}(\Phi(\delta\rho_\eta),\Phi(\delta\rho_\eta)) = \sum_{i\neq j} \; \frac{|\Phi(\delta\rho_\eta)_{ij}|^2}{\Phi(\rho_\eta)_{jj} f\norbraB{\Phi(\rho_\eta)_{ii}/\Phi(\rho_\eta)_{jj}}}  + \sum_i \frac{|\Phi(\delta\rho_\eta)_{ii}|^2}{\Phi(\rho_\eta)_{ii}}
		\end{align}
		where we used the formula for the Fisher information in the orthonormal basis of $\Phi(\rho_\eta)$ provided in~\cite{scandi2023quantum}.
		All the terms in the two sums are positive. Moreover, the second sum diverges at least as $\bigo{\eta^{-1}}$ as $\eta\rightarrow0$,  due to the contribution associated to the eigenvalue of $\ket{\psi_\eta}$. Hence, we can always find a $\eta$ small enough such that:
		\begin{align}
			K_{f,{\Phi}(\rho_\eta)}({\Phi}(\delta\rho_\eta),{\Phi}(\delta\rho_\eta))
			> \sup_{\lambda, \Tr{\delta\rho^2} = 1}\,K_{f,\rho_\lambda}(\delta\rho,\delta\rho)
			\geq
			K_{f,\rho_\eta}(\delta\rho_\eta,\delta\rho_\eta) \;,
		\end{align}
		contradicting the assumption that ${\Phi}$ contracts the Fisher metric for any two points in $\mathcal{P}_d^\circ$. This  proves the Positivity of $\Phi$. Moreover, choosing $\delta\rho\equiv\rho$ in Eq.~\eqref{eq:contractionContrast2}, one obtains:
		\begin{align}
			\Tr{\rho}= K_{f,\rho}(\delta\rho,\delta\rho) \geq
			K_{f,\Phi(\rho)}(\Phi(\delta\rho),\Phi(\delta\rho)) =\Tr{\Phi(\rho)}\,,
		\end{align}
		thanks to the property that $\J_f^{-1}\big|_\rho[\rho] = \idO$~\cite{scandi2023quantum}. This proves that $\Phi$ is trace non-increasing.
        \end{proof}
		
		Finally, regarding Cor.~\ref{cor:1}, it should be noticed that the only difference is that the tangent space of $\mathcal{S}_d$ is given by traceless Hermitian operators. Thus, whereas tangent vectors to $\mathcal{P}_d$ can be written as $\delta\rho=\sigma_1-\sigma_2$ for $\sigma_{1/2}\in\mathcal{P}_d$, in the case of $\mathcal{S}_d$ the same decomposition still holds, but with $\sigma_{1/2}\in\mathcal{S}_d$. Then, one can prove the existence of the perturbation $\delta\rho_\eta$ in the very same way as we did for $\mathcal{P}_d$. %Moreover, it should be noticed that the condition that $\Phi$ is trace preserving is not actually needed, as $\bra{\psi_0}\Phi(\frac{\rho_{\lambda^*}}{\Tr{\rho_{\lambda^*}}})\ket{\psi_0}=\frac{1}{\Tr{\rho_{\lambda^*}}}\bra{\psi_0}\Phi(\rho_{\lambda^*})\ket{\psi_0}=0$, so we can always choose a normalised state as $\sigma_2$ in Eq.~\eqref{eq:17A}, which still leads to a contradiction. 
	
	\section{2. Proof of Corollary~\ref{cor:2}}\label{app:2}
	We prove Corollary~\ref{cor:2}, which we repeat here for convenience:
	\begin{corollary*}\emph{
		Consider a  Hermitian preserving, linear map $\Phi: \mathcal{M}_d(\mathbb{C})\rightarrow\mathcal{M}_d(\mathbb{C})$ satisfying condition~\ref{cf:it:2} in Thm.~\ref{cf:theo:PiffContracts}. Consider any contrast function $H_g (\rho||\sigma) $. If for every two matrices $\rho$ and $\sigma$ on $\mathcal{P}_d^o$ it holds that:
		\begin{align}
			H_g (\rho||\sigma) \geq H_g (\Phi(\rho)||\Phi(\sigma))\,,
		\end{align}
		then $\Phi$ is a positive map. Moreover, if the same holds for $\Phi\otimes\idO_d$ then the map is completely positive.}
	\end{corollary*}
	\begin{proof}
		The proof of the corollary directly follows from the one of Thm.~\ref{cf:theo:PiffContracts}. Indeed, by contradiction, assume there exists a matrix $\sigma\in\mathcal{P}^\circ_d$ that  gets mapped outside of $\mathcal{P}_d$, and consider $\pi\in\mathcal{P}^\circ_d$ such that $\Phi(\pi)\in\mathcal{P}^\circ_d$. Analogously to the proof above, define $\rho_\lambda := (1-\lambda)\,\pi + \lambda \,\sigma$. For $\varepsilon$ small enough we can apply Thm.~\ref{cf:thm:Ruskai} to obtain:
		\begin{align}
			\sup_{\lambda, \Tr{\delta\rho^2} = 1}\,H_g(\rho_\lambda ||\rho_\lambda+\varepsilon\,\delta\rho) = \sup_{\lambda, \Tr{\delta\rho^2} = 1}\,\frac{\varepsilon^2}{2}\,K_{f,\rho_\lambda}(\delta\rho,\delta\rho)  +\bigo{\varepsilon^3}< \infty\,.
		\end{align}
		Then, following the steps presented for Thm.~\ref{cf:theo:PiffContracts}, we define a $\rho_\eta$ such that $\Phi(\rho_\eta)$ is of order $\eta$, where $\eta\ll 1$, and a $\delta\rho_\eta$ having  a positive finite contribution along the eigenvectors corresponding to $\eta$-eigenvalues. The positivity condition ensures that in the limit $\eta\rightarrow0$ the perturbed state is still in the interior of $ \mathcal{P}_d$ for any finite $\eta$, and $\varepsilon$ small enough. Then, $H_g(\Phi(\rho_\eta) ||\Phi(\rho_\eta+\varepsilon\,\delta\rho_\eta))$ scales as $\varepsilon^2/\eta$. Hence, one can find $\eta$ small enough such that:
		\begin{align}
			H_g(\Phi(\rho_\eta) ||\Phi(\rho_\eta+\varepsilon\,\delta\rho_\eta)) > \sup_{\lambda, \Tr{\delta\rho^2} = 1}\,H_g(\rho_\lambda ||\rho_\lambda+\varepsilon\,\delta\rho) \geq \;H_g(\rho_\eta ||\rho_\eta+\varepsilon\,\delta\rho_\eta) \,.
		\end{align}
		This gives the desired contradiction, concluding the proof.
	\end{proof}
	
\end{document}